\documentclass[final, 11pt]{article}

\usepackage{ltexpprt} \usepackage{amssymb} \usepackage{amsmath} \usepackage{appendix}

\begin{document}

\title{Scheduling Packets with Values and Deadlines in Size-bounded Buffers\thanks{Research is partially supported by NSF grant CCF-0915681.}}

\author{Fei Li\thanks{Department of Computer Science, George Mason University, Fairfax, VA 22030, USA. {\tt lifei@cs.gmu.edu}}}

\maketitle

%-------------------------------------------------------------------------

\begin{abstract}
Motivated by providing quality-of-service differentiated services in the Internet, we consider buffer management algorithms for network switches. We study a {\em multi-buffer model}. A network switch consists of multiple size-bounded buffers such that at any time, the number of packets residing in each individual buffer cannot exceed its capacity. Packets arrive at the network switch over time; they have values, deadlines, and designated buffers. In each time step, at most one pending packet is allowed to be sent and this packet can be from any buffer. The objective is to maximize the total value of the packets sent by their respective deadlines. A $9.82$-competitive online algorithm has been provided for this model (Azar and Levy. SWAT 2006), but no offline algorithms have been known yet. In this paper, We study the offline setting of the multi-buffer model. Our contributions include a few optimal offline algorithms for some variants of the model. Each variant has its unique and interesting algorithmic feature. These offline algorithms help us understand the model better in designing online algorithms.
\end{abstract}

\newpage

%-------------------------------------------------------------------------

\section{Introduction}

Motivated by providing quality-of-service differentiated services in the Internet, we consider buffer management algorithms for network switches. We study a {\em multi-buffer model}. A network switch consists of $m$ size-bounded buffers $Q_1$, $Q_2$, $\ldots$, $Q_m$; their sizes are denoted as $B_1, \ B_2, \ \ldots, \ B_m$ respectively. At any time, the number of packets residing in each individual buffer $Q_i$ cannot exceed its capacity $B_i$. Time is discretized into time steps. Packets arrive at the network switch over time and each packet $p$ has an integer arriving time (release time) $r_p \in \mathbb R^+$, a non-negative value $v_p \in \mathbb R^+$, an integer deadline $d_p \in \mathbb Z^+$, and a designated buffer $b_p \in \{Q_1, \ \ldots, \ Q_m\}$ that it can reside in. The deadline $d_p$ specifies the time by which $p$ should be sent. This model is preemptive such that the packets already existing in the buffers can be dropped at any time before they are transmitted. A dropped packet cannot be delivered any more. In each time step, at most one pending packet is allowed to be sent and this packet may be from any buffer. The objective is to maximize {\em weighted throughput}, defined as the total value of the packets transmitted by their respective deadlines.

The first QoS buffer management model is introduced in~\cite{AMRR05}. Since then, quite a few researchers have studied this model as well as other variants, mostly in the online settings~\cite{KLMPSS04, H01, CJST07, LSS05, EW07}. A well-studied model is called the {\em bounded-delay model}. In this model, there is only one buffer. Packets have integer release time, integer deadlines, and non-negative values. The objective is to maximize the total value of the packets sent by their deadlines. An implicit assumption on this model is the buffer's sufficiently large size. All released packets can be stored in the buffer before they are delivered or they get to expire. For the bounded-delay model, an optimal offline algorithm running in $O(n \log n)$ time has been proposed in~\cite{KLMPSS04}, where $n$ is the number of packets released. We call the bounded-delay model a {\em bounded-buffer model} in case the buffer size is enforced to be finite. The bounded-buffer model generalizes the bounded-delay model, if we allow the buffer size to be larger than any packet's {\em slack time}. (A packet's slack time is defined as the difference between its deadline and release time.) The bounded-buffer model is one variant of the {\em multi-buffer model} proposed by Azar and Levy~\cite{AL06}. A $9.82$-competitive online algorithm has been provided for this model~\cite{AL06}, but no offline algorithms have been known yet. In this paper, we study the offline setting of the multi-buffer model. Our contributions include a few optimal offline algorithms for some variants of the model. Each variant has its unique and interesting algorithmic feature. These offline algorithms help us understand the model better when we are designing online algorithms.

The variants and their algorithms' running complexities are summarized in Table~\ref{tbl:summary}. In the {\em uniform-value} setting, all packets have the same value. In the {\em non-uniform-value} setting, packets are allowed to have arbitrary values. (In designing offline algorithms, there is no difference between preemptive and non-preemptive settings.)

\begin{table}[h]
\begin{tabular}{|c|c|c|}
\hline \hline
& uniform-value setting & non-uniform-value setting \\ \hline
$m = 1$ & $\Omega(n \log \min\{B, \ n\})$ & $O(n^2)$ \\ \hline
$m > 1$ (packets sharing a common deadline) & $O(n \log n)$ & $O(n^2 \log n)$ \\ \hline \hline
\end{tabular}
\label{tbl:summary}
\caption{Summary of the running complexities of the optimal offline algorithms for some variants of the multi-buffer model. $n$ is the number of packets in the input sequence. For the bounded-buffer model, the buffer size is $B \in \mathbb Z^+$.}
\end{table}

%------------------------------------------------------------------------------

\section{The Bounded-buffer Model, $m = 1$}

Let OPT denote an optimal offline algorithm. Without loss of generality, we assume OPT is {\em non-idling}, that is, OPT sends a packet as long as the buffer is non-empty.

%------------------------------------------------------------------------------

\subsection{The uniform-value setting.}

In the uniform-value setting, all packets have the same `weight' and the objective is to maximize the number of packets delivered successfully. An optimal offline algorithm called DOS (which stands for `Deadline-Order-Sorting/Sending') works simply as follows.

\begin{algorithm}
All packets in the buffer are organized by their deadlines using an augmented red-black tree~\cite{CLRS01}. Upon each new arrival, we insert it into the packet queue in increasing order of deadlines. Let the current time be $t$. If the buffer is full or if more than $t' - t$ packets are to be sent by some deadline $t'$ (we call these cases `{\em tight}'), we drop the packet with the earliest deadline. In each time step, the earliest-deadline packet in the buffer is sent.
\end{algorithm}

\begin{Lemma}
For the bounded-buffer model in the uniform-value setting, there exists an optimal offline algorithm running in $O(n \log \min\{B, \ n\})$ time, where $n$ is the number of packets released and $B$ is the buffer size.
\label{lemma:singleoff1}
\end{Lemma}

\begin{proof}
We first prove DOS's correctness using a loop invariant. The loop invariant is: At any time, there exists a one-to-one mapping (injection) from each packet $q$ in OPT's buffer to a packet $j$ in DOS's buffer such that $d_q \le d_j$. Without loss of generality, we align the mappings such that an earlier-deadline packet in OPT's buffer maps to an earlier-deadline packet in DOS's buffer. For example, assume $q_1$ and $q_2$ in OPT's buffer map to $j_1$ and $j_2$ in DOS's buffer respectively. If $d_{q_1} < d_{q_2}$ but $d_{j_1} \ge d_{j_2}$, we swap the mappings and let $q_1$ map to $j_2$ and $q_2$ map to $j_1$. Note $d_{q_1} \le d_{q_2} \le d_{j_2} \le d_{j_1}$ and $d_{q_2} \le d_{j_2} \le d_{j_1}$.

This invariant holds before any packet is released. Let us assume it holds at time $t$. Consider a new arrival $p$ accepted by OPT. $p$ is either accepted by DOS or there exists a packet $j$ which is not mapped yet by any packet in OPT's buffer has a deadline $d_j \ge d_p$. (In this case, we can map $p$ in OPT's buffer to $j$ in DOS's buffer.) Otherwise, we can drop $j$ and accept $p$ or OPT's buffer is `tight' as well and OPT rejects $p$. In each time step, both OPT and DOS send one packet as long as their buffers are non-empty. Without loss of generality, we can assume OPT sends the earliest-deadline packet in its buffer. Thus, the loop invariant still holds after each step's deliveries. The loop invariant implies the correctness of the algorithm.

For each new arrival, it takes $O(\log \min\{B, \ n\})$ to insert $p$ into or drop $p$ out of the packet queue in DOS's buffer. The algorithm has an upper bound of running time $O(n \log \min\{B, \ n\})$. The proof is completed. $\Box$
\end{proof}

The following instance shows that no algorithm has a running complexity asymptotically better than $\Omega(n \log \min\{B, \ n\})$.

\begin{Example}
Assume $B \ge n$. All packets are released at the same time $0$. To identify whether all packets can be delivered successfully, we have to sort them by deadlines such that packets can be delivered in an earliest-deadline-first (EDF) manner. The lower bound of sorting $n$ numbers takes $\Omega(n \log n)$~\cite{CLRS01}.
\end{Example}

\begin{corollary}
Consider the bounded-buffer model in the uniform-value setting. If packets' deadlines are weakly increasing along with their release time, EDF is an optimal algorithms running in linear time. Specifically, EDF runs in an online manner.
\end{corollary}

%------------------------------------------------------------------------------

\subsection{The non-uniform-value setting.}

If $B \ge n$, the optimal offline algorithm~\cite{KLMPSS04} for the bounded-delay model applies on the bounded-buffer model and has a running time of $O(n \log n)$. We assume $B < n$. Fix an input sequence $\mathcal I$. We have the following algorithm.

\begin{algorithm}
We sort all packets in $\mathcal I$ in non-increasing value order, with ties broken in favor of the one with a later deadline. We start from a set of packets $S = \emptyset$. For each packet $j \in ({\mathcal I} \setminus S)$, we pick up $j$ in order and run EDF to examine whether all packets in $S \cup \{j\}$ can be delivered successfully by their respective deadlines. (Actually, we can start from the time $r_j$ to run EDF over the packets $S \cup \{j\}$ instead of from scratch; though this does not help to reduce the asymptotic running complexity.) If ``yes'', we update $S$ with $S \cup \{j\}$. For each examined packet $j$, no matter whether we insert $j$ into $S$ or not, we drop it out of $\mathcal I$. We examine all packets in $\mathcal I$ in order till $\mathcal I$ gets empty.
\label{alg:singlenon}
\end{algorithm}

\begin{Lemma}
For the bounded-buffer model in the non-uniform-value setting, there exists an optimal offline algorithm running in $O(n^2)$ time, where $n$ is the number of packets released.
\label{lemma:singleoff2}
\end{Lemma}

\begin{proof}
We claim that the schedule of $S$ we finally have from Algorithm~\ref{alg:singlenon} is optimal, based on the matroid property of this model. Consider a set of packets that can be delivered successfully by their deadlines in an EDF manner. Then, its any subset can be delivered successfully as well and the heredity property is satisfied. Also, in each time step, only one packet is allowed to send, and thus, the exchange property holds.

Let $|{\mathcal I}| = n$. Sorting packets in $\mathcal I$ takes $O(n \log n)$ time. The buffer has at most $B$ packets at any time, thus, each packet insertion (in increasing deadline order) takes $O(\log B)$ time. Running EDF over a set of packets $S \cup \{j\}$ takes time $|S| + 1 \le n$. For each packet $j$, examining $S \cup \{j\}$ of being successfully sent takes time $O(\log B + n)$. Thus, the total running time of the algorithm is $O(n \log n + n (n + \log B)) = O(n^2 + n \log B)$. Thus, our algorithm has a running time of $O(n^2)$. The proof is completed. $\Box$
\end{proof}

%------------------------------------------------------------------------------

\section{Scheduling Packets with a Common Deadline or Without Deadlines, $m > 1$}

Let OPT denote an optimal offline algorithm. Without loss of generality, we assume OPT is non-idling. In scheduling packets without deadlines, we assume all packets have a common deadline $r_{\max} + n$, where $r_{\max}$ is the largest release time. We also note that when there are no new arrivals, all packets already in the buffers can be sequentially delivered.

Let $P_i(t)$ denote the set of packets released at time $t$ targeting the buffer $Q_i$. Since each buffer $Q_i$ cannot accommodate more than $B_i$ packets at any time, we assume that for each $Q_i$, at any release time $t$, $|P_i(t)| \le B_i$. Let $Q_i(t)$ and $|Q_i(t)|$ denote the packet queue in the buffer $Q_i$ and its size, respectively.  Let $r^i_{\max}$ denote the largest release time of a packet targeting the buffer $Q_i$. Let $D$ be the common deadline.

%------------------------------------------------------------------------------

\subsection{The uniform-value setting.}

In the uniform-value setting, all packets have the same `weight' and the objective is to maximize the number of packets delivered successfully. Instead of directly targeting maximizing the total number of packets delivered, we tackle this variant from the perspective of minimizing the number of packets dropped. For each buffer, our idea is to calculate the number of buffer slots that we have to reserve in order to accept future arrivals (that is, minimizing the number of packets dropped due to `packet overflow'). This value indicates us the latest time that we have to deliver a packet from a buffer.

\begin{algorithm}
For each buffer $Q_i$, consider $P_i(t)$ in decreasing order of release time $t$. Define a variable $Z_i(t)$ to denote the number of buffer slots that are needed from the buffer $Q_i$ to accommodate packets released at/after time $t$. Set $Z_i(r^i_{\max}) = \max\{|P_i(r^i_{\max})|, \ D - r^i_{\max}\}$. In reverse order of release time, we calculate $Z_i(t) = \min\{B_i, \ Z_i(t') + |P_i(t)| - (t' - t)\}$, where $t'$ is the immediate next release time (of packets) after time $t$ for $Q_i$.

For each new arrival, if its designated buffer is full, drop the packet. Otherwise, append the packet to the queue. In each time step $t$, send any packet from the buffer $Q_i$ if $Z_i({\tilde t}) + |Q_i(t)| \ge B_i$, where ${\tilde t}$ is the immediate next release time of packets for the buffer $Q_i$. Ties are broken arbitrarily. If all buffers $Q_i$ have $Z_i({\tilde t}) + |Q_i(t)| < B_i$, choose any packet to send. We switch to another buffer to send a packet only if this buffer is empty or if another buffer $Q_i$ satisfies $Z_i({\tilde t}) + |Q_i(t)| \ge B_i$ at time $t$.
\label{alg:multiuniform}
\end{algorithm}

\begin{theorem}
In scheduling packets with the same value and same deadline, there exists an optimal offline algorithm running in $O(n \log n)$ time, where $n$ is the number of packets released.
\label{theorem:uniformoff}
\end{theorem}

\begin{proof}
We first show the correctness of Algorithm~\ref{alg:multiuniform} using the exchange argument. We call our algorithm TS (standing for `Tight Schedule'). Remember that all packets are with the same value and same deadline and TS accepts packets in a greedy manner for each buffer, thus, as long as OPT and TS schedule packets from the same buffer in each time step, they achieve the same throughput. Let $\cal O$ denote the set of packets sent by OPT. Let $t$ be the first time step in which OPT and TS deliver packets from different buffers. OPT sends a packet $q_1$ from a buffer $Q_1$ and TS sends a packet $p_2$ from a buffer $Q_2$. If $p_1 \notin \cal O$, it is fine for OPT sends $p_1$ in this time step such that $\cal O$ is updated with ${\cal O} \cup \{p_1\} \setminus \{q_1\}$. Here, we assume $p_1 \in \cal O$. Since we choose $Q_2$ to send a packet, one of the following cases must happen. At time $t$, we use $\hat t$ and $\tilde t$ to differentiate the two (possibly) distinct next release time of packets targeting buffers $Q_1$ and $Q_2$ respectively.

\begin{enumerate}
\item Assume $Z_1({\hat t}) + |Q_1(t)| < B_1$ and $Z_2({\tilde t}) + |Q_2(t)| < B_2$. In this case, delivering either $p_1$ or $q_1$ will not result packet overflow for both buffers $Q_1$ and $Q_2$. Thus, OPT can be changed to choose $Q_2$ to send a packet.

\item Assume $Z_1({\hat t}) + |Q_1(t)| < B_1$ and $Z_2({\tilde t}) + |Q_2(t)| \ge B_2$. In this case, if TS does not choose $Q_2$ to send a packet, one packet released at time ${\tilde t}$ or future will not be delivered successfully. Let this packet be $p$. Then, among all packets in $Q_2$'s current buffer and those packets released later targeting $Q_2$, one of them must not be in $\cal O$. Otherwise, OPT will choose $Q_2$ to send a packet to avoid $Q_2$'s packet overflow. Assume the packet sending sequence since time $t$ for OPT is $q_1, \ \ldots, \ p_1, \ \ldots$. We modify the sequence for OPT as $p_1, \ \ldots, \ p, \ \ldots$ and update $\cal O$ as ${\cal O} \cup \{p\} \setminus \{q_1\}$. Since $p_1$ is delivered in this time step, there exists an extra buffer slot (compared with that of the unmodified OPT which does not send $p_1$ for step $t$) to accommodate $p$ in the buffer $Q_2$ and thus, the new packet sequence is feasible. After our modification, OPT's total gain is not reduced and OPT chooses the same queue as TS does to send a packet in this time step.

\item Assume $Z_1({\hat t}) + |Q_1(t)| \ge B_1$ and $Z_2({\tilde t}) + |Q_2(t)| \ge B_2$. In this case, delivering either $p_1$ or $q_1$ will result packet overflow for the other buffer. Thus, with the same analysis as the above case, OPT can be changed to choose $Q_2$ to send a packet.
\end{enumerate}

We then show the running time of Algorithm~\ref{alg:multiuniform}. Sorting all distinct release time for each buffer takes $O(n \log n)$ time. Calculating the variables $Z_i(t)$ takes linear time $O(n)$. For each time $t$, we identify the buffer to send a packet and this takes time $O(m)$. In total, the running complexity of our algorithm is $O(n \log n)$. The proof is completed. $\Box$
\end{proof}

The proof of Theorem~\ref{theorem:uniformoff} immediately implies the following corollary.

\begin{corollary}
In scheduling packets with the same value and same deadline, Algorithm~\ref{alg:multiuniform} provides a way to identify whether a set of packets can be delivered successfully.
\end{corollary}

%------------------------------------------------------------------------------

\subsection{The non-uniform-value setting.}

We realize that when each buffer size is large enough, the multi-buffer model is same as the bounded-delay model since all arriving packets can be accommodated in the buffers. Hence, we have two trivial results on the non-uniform-value setting.

\begin{lemma}
For the multi-buffer model, if all buffers have their sizes larger than the maximum slack of a packet targeting at them, the multi-buffer model is same as the bounded-delay model. An optimal offline algorithm running in time $O(n \log n)$ exists, where $n$ is the number of packets released.
\end{lemma}

\begin{corollary}
For the multi-buffer model, if there is no future arrivals, there exists an optimal offline algorithm sending the packets in the buffers, running in $O(n \log n)$ time, where $n$ is the number of packets in the current buffers.
\label{coro:multi}
\end{corollary}

In scheduling weighted packets sharing a common deadline, our idea is to combine Algorithm~\ref{alg:singlenon} and Algorithm~\ref{alg:multiuniform}. We note that this variant is a matroid as well (this can be verified easily as in the proof of Lemma~\ref{lemma:singleoff2}). Then a greedy algorithm  scheduling packets with more values is optimal. Let $S$ be a set of packets we decide to send. Initially, $S$ is empty. We order packets in decreasing order of values. Then, we examine packets one by one, as long as the new one and those already selected packets can be delivered by the common deadline, we add this new packet into $S$. Otherwise, we drop this newly considered packet. There is a questions unsolved: How do we identify whether a set of selected packets can be delivered as they belong to multiple buffers at different time? We apply the idea of Algorithm~\ref{alg:multiuniform}, specifically, the result of Corollary~\ref{coro:multi}.

\begin{algorithm}
Fix an input instance $\cal I$. We sort all packets in $\mathcal I$ in non-increasing value order. We start from a set of packets $S = \emptyset$. For each packet $j \in ({\mathcal I} \setminus S)$, we pick up $j$ in order and examine whether all packets in $S \cup \{j\}$ can be delivered successfully. (See below.) If `yes', we update $S$ with $S \cup \{j\}$. For each examined packet $j$, no matter whether we insert $j$ into $S$ or not, we drop it out of $\mathcal I$. We examine all packets in $\mathcal I$ in order till $\mathcal I$ gets empty.

Let $P'_i(t)$ denote a subset of selected packets ($S$) which are released at time $t$ targeting the buffer $Q_i$. For each buffer $Q_i$, consider $P_i(t)$ in decreasing order of release time $t$. In reverse order of release time, we calculate $Z_i(t) = \min\{B_i, \ Z_i(t') + |P_i(t)| - (t' - t)\}$, where $t'$ is the immediate next release time (of packets) after time $t$ for $Q_i$.

For each new arrival, if its designated buffer is full, drop the packet and return `no'. Otherwise, append the packet to the queue. In each time step $t$, send any packet from the buffer $Q_i$ if $Z_i({\tilde t}) + |Q_i(t)| \ge B_i$, where ${\tilde t}$ is the immediate next release time of packets for the buffer $Q_i$. Ties are broken arbitrarily. If all buffers $Q_i$ have $Z_i({\tilde t}) + |Q_i(t)| < B_i$, choose any packet to send. We switch to another buffer to send a packet only if this buffer is empty or if another buffer $Q_i$ satisfies $Z_i({\tilde t}) + |Q_i(t)| \ge B_i$ at time $t$.
\label{alg:multinon}
\end{algorithm}

\begin{theorem}
In scheduling packets with the same deadline, there exists an optimal offline algorithm running in $O(n^2 \log n)$ time, where $n$ is the number of packets released.
\label{theorem:nonuniformoff}
\end{theorem}

\begin{proof}
The correctness of Algorithm~\ref{alg:multinon} depends on the matroid property of this variant and Corollary~\ref{coro:multi}.

We then show the running time of Algorithm~\ref{alg:multinon}. Sorting all distinct release time for each buffer takes $O(n \log n)$ time. Calculating the variables $Z_i(t)$ takes linear time $O(n)$. For each time $t$, we identify the buffer to send a packet and this takes time $O(m)$. In total, the running complexity of our algorithm in examining one packet is $O(n \log n)$. Thus, the total running time of Algorithm~\ref{alg:multinon} is $O(n^2 \log n)$. The proof is completed. $\Box$
\end{proof}

%------------------------------------------------------------------------------

\section{Conclusion}

In this paper, we design offline algorithms for some variants of the multi-buffer model. We show that if the number of buffers is restricted to $1$ or if all packets share a common deadline, some efficient offline algorithms can be developed. However, for the general case of the multi-buffer model, the constraints from the buffer sizes, packets' deadlines and packets' values complicate this packet scheduling problem. An optimal offline algorithm for the general multi-buffer model is being under developed.

%------------------------------------------------------------------------------

\bibliographystyle{plain}
\bibliography{../../buffer}

%------------------------------------------------------------------------------

\end{document}